\documentclass[12pt,onesided]{article}

\usepackage[leqno]{amsmath}
\usepackage{amsthm}
\usepackage{amssymb}
\usepackage{amsfonts}
\usepackage{array}
\usepackage{hyperref}

\newcommand{\R}{\mathbb{R}}

\numberwithin{equation}{section}

\newtheorem{lemma}{LEMMA}[section]

\def\bi{\begin{itemize}}
\def\ei{\end{itemize}}

\begin{document}

\title{A shorter proof of Lemma A.6}
\author{Tom Fischer\thanks{Institute of Mathematics, University of Wuerzburg, 
Campus Hubland Nord, Emil-Fischer-Strasse 30, 97074 Wuerzburg, Germany.
Tel.: +49 931 3188911.
E-mail: {\tt tom.fischer@uni-wuerzburg.de}.
}\\
University of Wuerzburg}
\date{\today}

\maketitle

\begin{abstract}
For the convenience of readers of the article {\em No-arbitrage pricing under systemic risk: accounting for cross-ownership} (Fischer, 2012), a full proof of Lemma A.5 and a shorter proof of Lemma A.6 of that paper
are provided. 
\end{abstract}


\section{Lemma A.5}

Lemma A.5 provides a method of splitting a positive number (or an interval) into
a pre-specified number of ordered summands (or subintervals) where we demand that, 
while the order increases, the summands (or subintervals) have pre-specified positive sizes 
(or pre-specified lengths) for as long as possible. 
While this formulation sounds trivial (the `algorithm' for solving this problem certainly is trivial), 
the resulting formula \eqref{capcons} is possibly not
directly obvious at first sight, especially if the previous formulation is not given.

\paragraph*{Lemma A.5.} 
{\em
For $x\in\R$, $m\in\{1,2,\ldots\}$, and $y^1,\ldots,y^m\in\R_0^+$,
\begin{eqnarray}
\label{capcons}
x & = & \min\left\{y^1,x\right\}
\; + \; \sum_{j=1}^{m-1}\;\; \min\left\{y^{j+1},\left(x-\sum_{i=1}^{j}y^i\right)^+\right\} 
\; + \; \left(x-\sum_{i=1}^{m}y^i\right)^+ .
\end{eqnarray}
}

\begin{proof}
The case $x \leq y^1$ is clear. The case $x \geq \sum_{i=1}^{m}y^i$ is clear, because then 
\begin{equation}
x - \sum_{i=1}^{j}y^i  \; \geq  \; \sum_{i=j+1}^{m}y^i  \; \geq  \; y^{j+1} \quad (j\in\{0,\ldots,m-1\}) .
\end{equation}
Excluding the two cases above, one must have $m \geq 2$ and
\begin{equation}
\label{1.3}
y^1 \; < \; x \; < \; \sum_{i=1}^{m}y^i .
\end{equation}
By the assumptions, there exists now an $i_0\in\{1,\ldots,m-1\}$ with
\begin{equation}
\label{1}
 \sum_{i=1}^{i_0}y^i  \; <  \;  x  \; \leq  \;  \sum_{i=1}^{i_0+1}y^i  ,
\end{equation}
which is equivalent to
\begin{equation}
\label{1.5}
0  \; <  \;  x - \sum_{i=1}^{i_0}y^i \; \leq  \;  y^{i_0+1} .
\end{equation}
From the left inequality in \eqref{1} one obtains
\begin{equation}
\label{1.6}
x - \sum_{i=1}^{j}y^i  \; >  \; \sum_{i=j+1}^{i_0}y^i  \; \geq  \; y^{j+1} \quad (j\in\{0,\ldots,i_0-1\}) .
\end{equation}
From the right inequality of \eqref{1}, one obtains
\begin{equation}
\label{1.7}
x - \sum_{i=1}^{i_0+1}y^i  \; \leq  \;  0 .
\end{equation}
Applying \eqref{1.3}, \eqref{1.6}, \eqref{1.5} and \eqref{1.7} to the right hand side of \eqref{capcons}, 
we obtain
\begin{equation}
y^1 + \ldots + y^{i_0} + \left(x - \sum_{i=1}^{i_0}y^i\right) \; = \; x .
\end{equation}
\end{proof}


\section{Lemma A.6}

The next lemma is a result which gives a condition under which the difference of two numbers
which have been split into the same amount of summands (of which some can be zero) according to Lemma 5
can be expressed as the sum of the absolute values of the differences of their summands.

\begin{lemma}
\label{lemma_2}
Assume $x_1,x_2\in\R$ where 
\begin{equation}
\label{y}
x_1 \; \geq \; x_2 
\end{equation}
and $y^i_1,y^i_2\in\R_0^+$ ($i=1,\ldots,m$) where
\begin{equation}
\label{z1}
y^i_1 \; \geq \; y^i_2 \quad (i=1,\ldots,m)
\end{equation}
and
\begin{equation}
\label{z2}
x_1-x_2 \; \geq \; \sum_{i=1}^{m} \left(y^i_1 - y^i_2\right) .
\end{equation}
Then, the following equation holds:
\begin{eqnarray}
\label{stuffnn}
x_1-x_2 & = & 
\left|\min\left\{y^1_1,x_1\right\} - \min\left\{y^1_2,x_2\right\}\right|\\
\nonumber &  & + \; 
\sum_{j=1}^{m-1}\;\; \left|\min\left\{y^{j+1}_1,\left(x_1-\sum_{i=1}^{j}y^i_1\right)^+\right\}\right. \\
\nonumber &  & \qquad\qquad 
- \; \left.\min\left\{y^{j+1}_2,\left(x_2-\sum_{i=1}^{j} y^i_2\right)^+\right\}\right| \\
\nonumber &  & + \left|\left(x_1-\sum_{i=1}^{m} y^i_1\right)^+ 
- \left(x_2-\sum_{i=1}^{m} y^i_2\right)^+\right| .
\end{eqnarray}
\end{lemma}

\begin{proof}
\eqref{z2} together with \eqref{z1} implies
\begin{equation}
x_1-x_2 \; \geq \; \sum_{i=1}^{j} \left(y^i_1 - y^i_2\right) \quad (j=1,\ldots,m) ,
\end{equation}
and therefore 
\begin{equation}
\label{yz}
x_1 - \sum_{i=1}^{j} y^i_1  \; \geq \; x_2 - \sum_{i=1}^{j} y^i_2 \quad (j=1,\ldots,m) .
\end{equation}
\eqref{y}, \eqref{z1} and \eqref{yz} imply that all differences on the right hand side of
\eqref{stuffnn} are non-negative. We can therefore apply \eqref{capcons}
(with $x \,\widehat{=}\, x_{1/2}$ and $y^i \,\widehat{=}\, y^i_{1/2}$) and obtain the result.
\end{proof}

\noindent
Obviously, $y^i_1 = y^j_1 = y^i_2 = y^j_2 \geq 0$ for all $i,j\in\{1,\ldots,m\}$ would satisfy the conditions
of the lemma.\\

Lemma A.6 of Fischer (2012) follows now by noting that, without loss
of generality, $y^1 \geq y^2$ can be assumed in the proof of that lemma. Therefore,
Lemma \ref{lemma_2} (with $x_j \,\widehat{=}\,x+y^j$ and $y^i_j \,\widehat{=}\, \psi^i(y^j)$ 
for $i=1,\ldots,m$ and $j=1,2$) can be applied.



\end{document}